\documentclass{ceurart}

\usepackage{amssymb}
\usepackage{amsmath}
\usepackage{amsthm}
\usepackage{bm, bbm}
\usepackage{caption}
\usepackage[capitalize]{cleveref}
\usepackage{scalerel}
\usepackage{stackengine}
\usepackage{subcaption}
\usepackage{tikz}
\usepackage{xspace}
\usepackage{algorithm}
\usepackage{algpseudocode}

\newtheorem{theorem}{Theorem}

\renewcommand{\vec}[1]{\bm{#1}}
\newcommand{\mat}[1]{\bm{#1}}

\newcommand{\abs}[1]{\lvert #1 \rvert}
\newcommand{\ind}[1]{\mathbbm{1}\lbrace #1 \rbrace}
\newcommand{\norm}[1]{\lVert #1 \rVert}

\definecolor{vir0}{RGB}{68,1,84}
\definecolor{vir1}{RGB}{72,24,106}
\definecolor{vir2}{RGB}{71,45,123}
\definecolor{vir3}{RGB}{66,64,134}
\definecolor{vir4}{RGB}{59,82,139}
\definecolor{vir5}{RGB}{51,99,141}
\definecolor{vir6}{RGB}{44,114,142}
\definecolor{vir7}{RGB}{38,130,142}

\newcommand{\Ltorch}{L_{\mathrm{torch}}}
\newcommand{\Lmin}{L_{\mathrm{min}}}
\newcommand{\QUBO}{\textsc{Qubo}\@\xspace}
\newcommand{\SC}{\textsc{SetCover}\@\xspace}
\newcommand{\TP}{\textsc{TorchPlacement}\@\xspace}

\newcommand{\BB}{\mathbb{B}}
\newcommand{\cS}{\mathcal{S}}
\newcommand{\cT}{\mathcal{T}}

\newcommand*{\defeq}{\mathrel{\vcenter{\baselineskip0.5ex\lineskiplimit0pt\hbox{\scriptsize.}\hbox{\scriptsize.}}}%
	=}
\newcommand*{\eqdef}{=\mathrel{\vcenter{\baselineskip0.5ex\lineskiplimit0pt\hbox{\scriptsize.}\hbox{\scriptsize.}}}}
\newcommand\equalhat{\mathrel{\stackon[1.5pt]{=}{\stretchto{%
				\scalerel*[\widthof{=}]{\wedge}{\rule{1ex}{3ex}}}{0.5ex}}}}

\DeclareMathOperator*{\argmin}{arg\,min}

\begin{document}

\copyrightyear{2023}
\copyrightclause{Copyright for this paper by its authors.
  Use permitted under Creative Commons License Attribution 4.0
  International (CC BY 4.0).}

\conference{LWDA'23: Lernen, Wissen, Daten, Analysen.
  October 09--11, 2023, Marburg, Germany}

\title{Efficient Light Source Placement\\using Quantum Computing}

\author[1]{Sascha M\"ucke}[%
orcid=0000-0001-8332-6169,
email=sascha.muecke@tu-dortmund.de,
url=https://smuecke.de/
]
\address[1]{Lamarr Institute, TU Dortmund University, Dortmund, Germany}

\author[2]{Thore Gerlach}[%
orcid=0000-0001-7726-1848,
email=thore.gerlach@iais.fraunhofer.de
]
\address[2]{Fraunhofer IAIS, Sankt-Augustin, Germany}

\begin{abstract}
\textsf{NP}-hard problems regularly come up in video games, with interesting connections to real-world problems.
In the game Minecraft, players place torches on the ground to light up dark areas.
Placing them in a way that minimizes the total number of torches to save resources is far from trivial.
In this paper, we use Quantum Computing to approach this problem.
To this end, we derive a QUBO formulation of the torch placement problem, which we uncover to be very similar to another \textsf{NP}-hard problem.
We employ a solution strategy that involves learning Lagrangian weights in an iterative process, adding to the ever growing toolbox of QUBO formulations.
Finally, we perform experiments on real quantum hardware using real game data to demonstrate that our approach yields good torch placements.
\end{abstract}


\maketitle

\section{Introduction}

Intriguing scientific problems pop up when you would rather take your mind off them for a while, e.g., when playing video games --- although, of course, it comes as no complete surprise to encounter them in this domain.
Video games, by nature, pose a wide range of problems that players are expected to solve in order to progress, such as slaying enemies, traversing mazes or solving logical puzzles.
Some games have even been shown to be \textsf{NP}-complete \cite{kaye2000}.
It is, however, a mild surprise to encounter them in Minecraft, which is an exceptionally serene open-world sandbox game developed by Mojang and first published around 2010\footnote{\url{https://www.minecraft.net/}}.
In Minecraft, players can roam around a vast world containing mountains, oceans and caves made up entirely of cubic blocks, giving the game its unique and slightly retro visual quality (see \cref{fig:screenshots}).
The game has been praised for fostering creativity in children (and adults) and has even been used as a tool for teaching \cite{nebel2016}.

\begin{figure}[t]
	\centering
	\begin{subfigure}[t]{.45\textwidth}
		\centering
		\includegraphics[width=\textwidth]{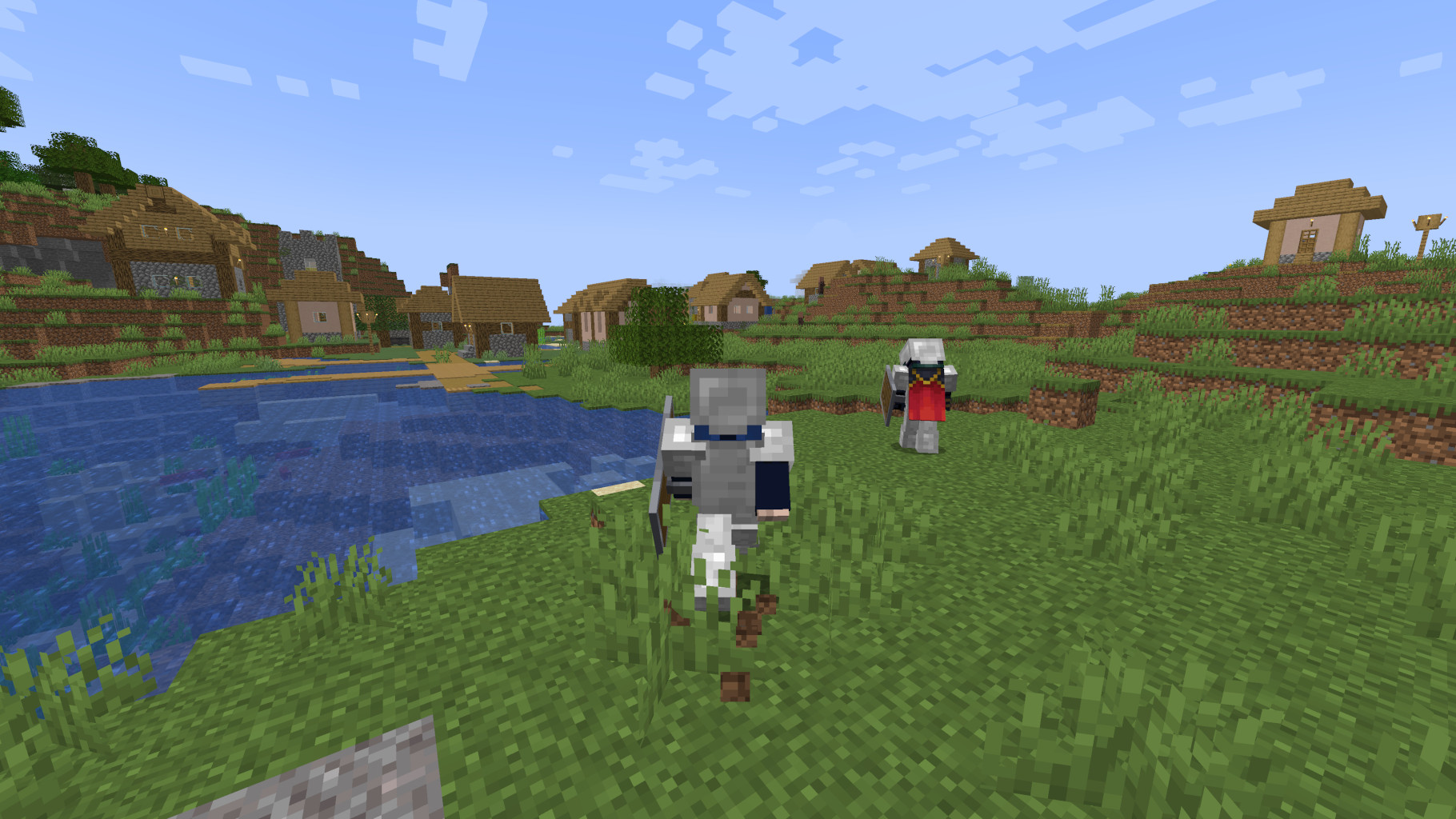}
		\subcaption{Players approaching a village}
		\label{fig:screenshot1}
	\end{subfigure}%
	\hspace*{5mm}
	\begin{subfigure}[t]{.45\textwidth}
		\centering
		\includegraphics[width=\textwidth]{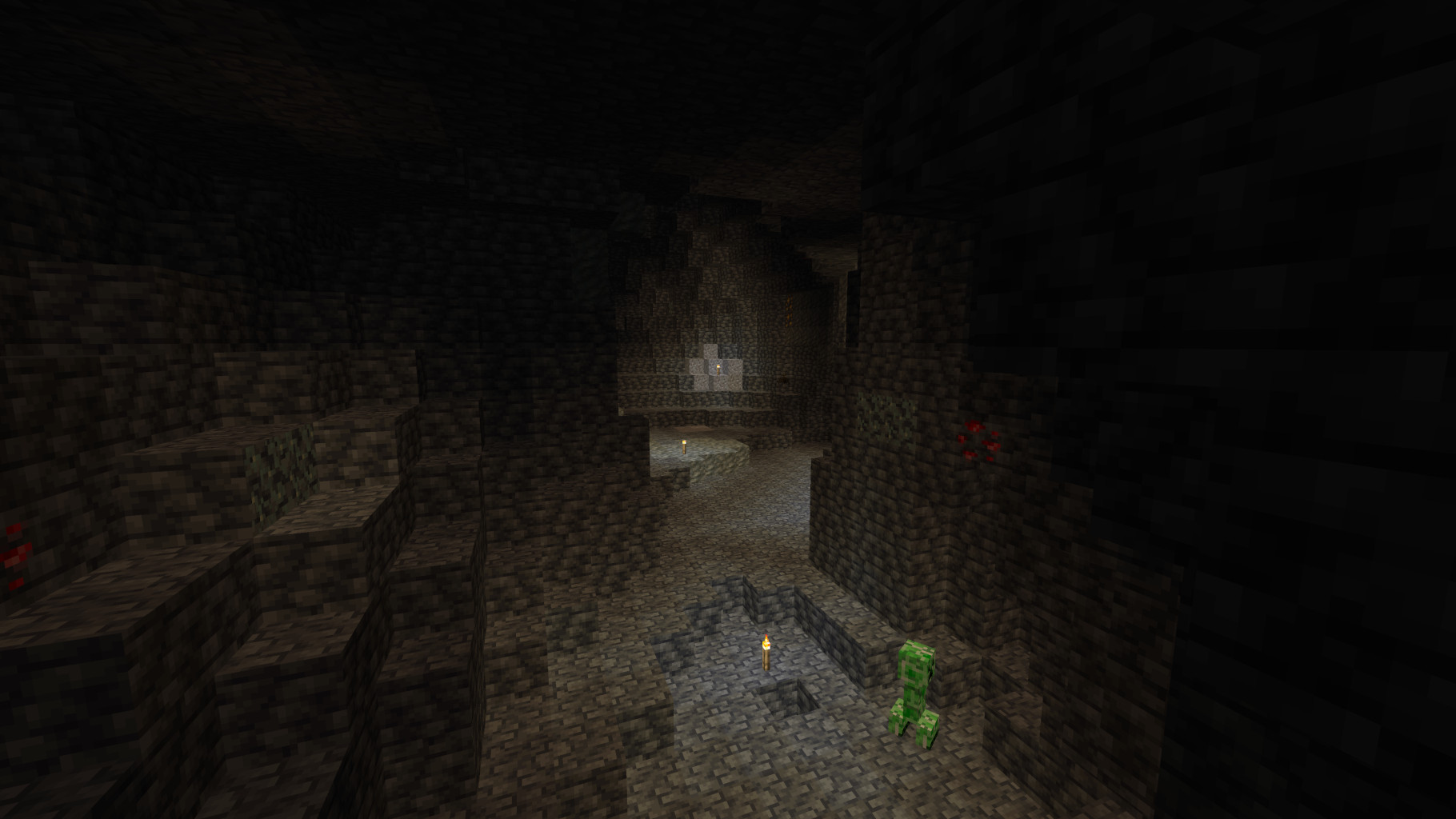}
		\subcaption{Dimly lit cave with hostile creature}
		\label{fig:screenshot2}
	\end{subfigure}
	\caption{Screenshots of Minecraft gameplay: In \cref{fig:screenshot1} two players are running towards a village through a landscape with a river on the left.
		The entire game world is made up of cubes called ``blocks''.
		\Cref{fig:screenshot2} shows a cave with a few torches that illuminate the floor partly.
		In the foreground, a hostile creature called ``Creeper'' is searching for players to attack. Such creatures spawn on the surface of blocks that are below a certain light level.}
	\label{fig:screenshots}
\end{figure}

\begin{figure}[b]
	\centering
	\begin{subfigure}[t]{0.3\textwidth}
		\centering
		\includegraphics[width=\textwidth]{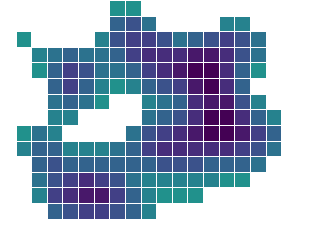}
		\subcaption{$20\times 15$}
	\end{subfigure}
	\hfill
	\begin{subfigure}[t]{0.3\textwidth}
		\centering
		\includegraphics[width=\textwidth]{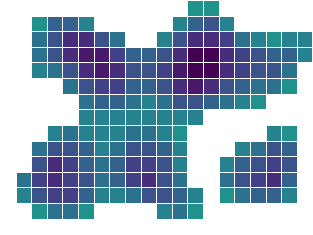}
		\subcaption{$20\times 15$ }
	\end{subfigure}
	\hfill
	\begin{subfigure}[t]{0.3\textwidth}
		\centering
		\includegraphics[width=\textwidth]{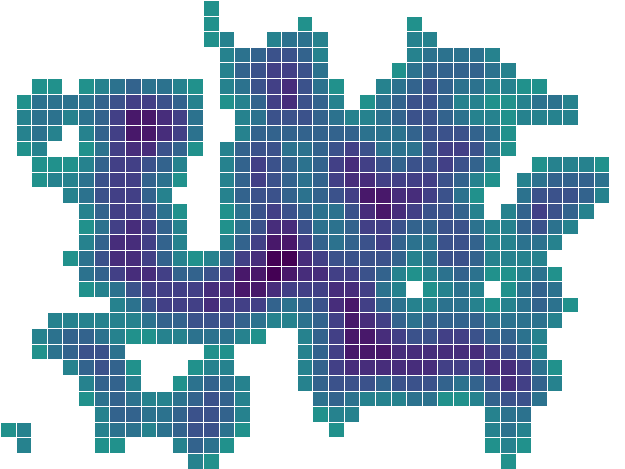}
		\subcaption{$40\times 30$}
	\end{subfigure}
	
	\begin{subfigure}[t]{0.3\textwidth}
		\centering
		\includegraphics[width=\textwidth]{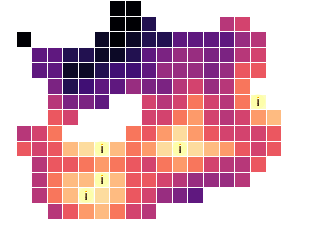}
		\subcaption{same as above, with torches}
	\end{subfigure}
	\hfill
	\begin{subfigure}[t]{0.3\textwidth}
		\centering
		\includegraphics[width=\textwidth]{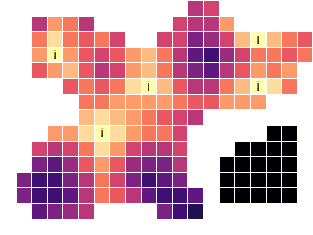}
		\subcaption{same as above, with torches}
	\end{subfigure}
	\hfill
	\begin{subfigure}[t]{0.3\textwidth}
		\centering
		\includegraphics[width=\textwidth]{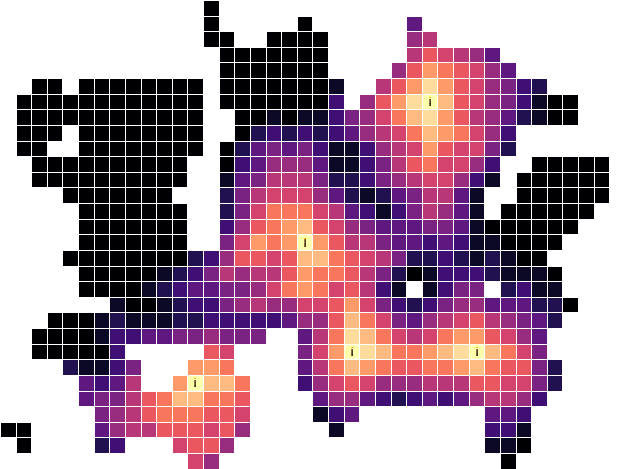}
		\subcaption{same as above, with torches}
	\end{subfigure}
	
	\caption{Examples of heightmaps used as input data for our experiments.
		In the top row, lighter color indicates higher elevation $z$.
		In the bottom row, lighter color indicates light level, and the little torch symboles represent torches placed on the respective tiles.
		White tiles are walls ($z=\infty$).
		The data is randomly generated from 2D Perlin noise.}
	\label{fig:heightmap-examples}
\end{figure}

However, in the darkness of a cave, hostile monsters can spawn on top of blocks that are insufficiently lit, either by the sun or a light source such as a torch or lantern.
These monsters will attack the player, who therefore wants to light up dark areas as best as they can while exploring.
At the same time, building torches requires valuable resources, which is why the player wants to use them sparingly, which leaves us with an interesting constrained optimization problem:
Where should we place torches to properly light up an area using as few torches as possible?

The first key observation is that, given the discrete nature of Minecraft, solutions to this problem can be described by binary variables $x_s\in\lbrace 0,1\rbrace\eqdef\BB$, where $x_s=1$ indicates that we should place a torch at location $s$.
Problems of this kind can be approached with Quantum Computing, which has been a promising focus of research in recent years.
Quantum Computers use \emph{Qubits} to perform computations, which can take values $0$ and $1$, but also be in a \emph{superposition} of both at once.
In particular, Quantum Annealers are being built that can be used to find the ground state of Ising models by means of quantum tunneling effects \cite{kadowaki1998,farhi2000}.
This model is equivalent to a class of quadratic binary optimization problems called \QUBO:
Given an upper triangular weight matrix $\mat Q\in\mathbb{R}^{n\times n}$, we want to find a binary vector $\vec x^*$ such that \begin{align}
	\vec x^* &= \underset{\vec x\in\BB^n}{\arg\min} ~f_{\mat Q}(\vec x) \nonumber\\
	\text{with}\quad f_{\mat Q}(\vec x)&\defeq \vec x^\top\mat Q\vec x = \sum_{1\leq i\leq j\leq n}Q_{ij}x_ix_j\label{eq:qubo}\;,
\end{align}%
where $f_{\mat Q}$ is called the \emph{energy} function.
This problem is \textsf{NP}-hard, and numerous problems have been re-formulated to be solvable in the form of \QUBO, ranging from economics and finance \cite{laughhunn1970,hammer1971,biesner2022} over satisfiability \cite{kochenberger2005}, resource allocation and routing problems \cite{neukart2017,stollenwerk2019} to machine learning \cite{bauckhage2018,mucke2019a,date2020,mucke2023,piatkowski2022,bauckhage2019}.
The promise of efficient solvability on Quantum Hardware has renewed and intensified research efforts for finding \QUBO formulations of various problems.

In this paper, we derive a formulation for the \TP problem.
Due to the large number of constraints, our method involves a learning procedure for Lagrangian weights. This is performed in an iterative fashion on a real quantum computer, similar to \cite{yonaga2020}.
We demonstrate our method in a range of experiments, using both generated and actual game data from a Minecraft world.
Further, we discuss the surprising connection to another optimization problem that we uncovered in the process of working on this article.

In \cref{sec:problem}, we give a formal description of the \TP problem.
In \cref{sec:qubo}, we develop a \QUBO formulation that attempts to solve it, combining quantum-enhanced optimization with an iterative learning scheme.
In \cref{sec:experiments}, we use our method to solve a range of example instances of \TP and discuss our observations.
\Cref{sec:related} puts our work in the context of similar work and related problems.
Finally, in \cref{sec:conclusion} we summarize our findings.

\section{Problem Statement}
\label{sec:problem}

As input, we are given a heightmap of a room, consisting of a rectangular grid of tiles, as shown in \cref{fig:heightmap-examples}.
The grid has integer width $w$ and height $h$.
Additionally, each grid site $s=(i,j)\in [h]\times[w]$ has an elevation value $z(s)\in\mathbb{N}_0\cup\lbrace\infty\rbrace$, where $\mathbb{N}_0$ denotes the set of natural numbers including $0$.
We use the notation $[n]\defeq\lbrace 1,\dots,n\rbrace$.
Intuitively, $z(s)$ is the ``floor level'' of $s$.
If $z(s)=\infty$, the tile is considered to be a wall.
Let $\cS=\lbrace s\in[h]\times[w] ~\vert ~z(s)<\infty\rbrace$ the set of all grid sites that are not walls.

Intuitively, the height map describes the top-down view of a 3-dimensional room consisting of cubic \emph{blocks} that can be either empty or non-empty (see \cref{fig:real_caves}).
Assuming a discrete coordinate system, each block has a location $(i,j,k)^\top\in\mathbb{Z}^3$.
If $z(i,j)=r$, then $\forall r'<r\leq r'': ~(i,j,r'')^\top$ is empty and $(i,j,r')^\top$ is non-empty.
Let the predicate $\mathsf{E}(\vec{p})\equiv\top$ (true) if and only if the block at $\vec{p}$ is empty.
Two blocks at $\vec{p}$ and $\vec{p}'$ are considered neighbors if $\norm{\vec{p}-\vec{p}'}_1=1$, i.e., if they share a face.
The set of empty neighbors of $\vec{p}$ is $\mathcal{N}(\vec{p})=\lbrace \vec p+\vec u \,\vert\, \vec u\in\mathbb{Z}^3, \,\lVert\vec u\rVert_1=1, \,\mathsf{E}(\vec p+\vec u)\rbrace$.
Finally, the distance between two sites $s=(i,j)$ and $t=(i',j')$, denoted by $d(s,t)$, is the shortest path from $(i,j,z(s))^\top$ to $(i',j',z(t))^\top$ moving through empty neighboring blocks:
\begin{equation}\label{eq:distance}
    d(s,t)=\begin{cases}
        0 &\text{if } s = t, \\
        1+\min_{s'\in\mathcal{N}(s)}d(s',t) &\text{otherwise}
    \end{cases}
\end{equation}

\noindent If there is no path between the two blocks, we set $d(s,t)=\infty$.
Examples of heightmaps with distances are shown in \cref{fig:distance-example}.

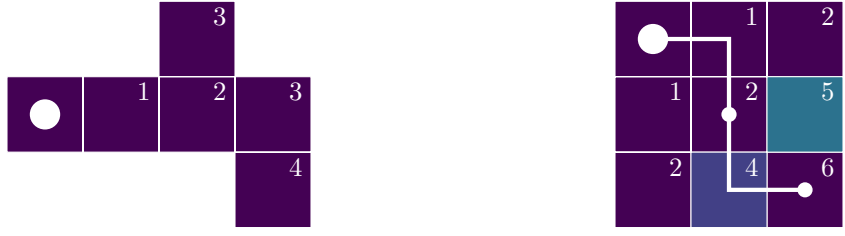
\begin{figure}
	\centering
	\begin{tikzpicture}[every node/.style={white}]
		\draw[draw=white, fill=vir0]
			(0,0) rectangle ++(1, 1)
			(1,0) rectangle ++(1, 1) node at (1.8, 0.8) {$1$}
			(2,0) rectangle ++(1, 1) node at (2.8, 0.8) {$2$}
			(2,1) rectangle ++(1, 1) node at (2.8, 1.8) {$3$}
			(3,0) rectangle ++(1, 1) node at (3.8, 0.8) {$3$}
			(3,-1) rectangle ++(1, 1) node at (3.8, -0.2) {$4$};
		\draw[draw=none, fill=white] (0.5, 0.5) circle (0.2);
	
		\draw[draw=white, fill=vir0]
			(8,1) rectangle ++(1,1)
			(9,1) rectangle ++(1,1) node at (9.8,1.8) {$1$}
			(10,1) rectangle ++(1,1) node at (10.8,1.8) {$2$}
			(8,0) rectangle ++(1,1) node at (8.8,0.8) {$1$}
			(9,0) rectangle ++(1,1) node at (9.8,0.8) {$2$}
			(8,-1) rectangle ++(1,1) node at (8.8,-0.2) {$2$}
			(10,-1) rectangle ++(1,1) node at (10.8,-0.2) {$6$};
		\draw[draw=white, fill=vir3]
			(9,-1) rectangle ++(1,1) node at (9.8,-0.2) {$4$};
		\draw[draw=white,fill=vir6]
			(10,0) rectangle ++(1,1) node at (10.8,0.8) {$5$};
		\draw[draw=none, fill=white] (8.5,1.5) circle (0.2);
		\draw[draw=white, ultra thick]
			(8.5,1.5) -- ++(1,0) -- ++(0,-2) -- ++(1,0);
		\draw[draw=none, fill=white]
			(9.5,0.5) circle (0.1)
			(10.5,-0.5) circle (0.1);
	\end{tikzpicture}
	\caption{Examples of distances in heightmaps: The number in the top right corner is the distance to the tile with the white circle. The distance is the shortest path through empty space moving from neighbor to neighbor.
	\textbf{Left:} All tiles have the same elevation. \textbf{Right:} The tile with distance 5 has elevation 2, the tile with distance 4 has elevation 1, all the others have elevation 0. A shortest path from the top left to the bottom right tile is indicated by the white line. Changes in elevation are drawn as smaller circles. The shortest path leads over the tile with elevation 1. The number of discrete steps taken gives the distance.}
	\label{fig:distance-example}
\end{figure}

Given a heightmap,  our objective is to place torches on some of the tiles in such a way that all tiles have a minimal light level of $\Lmin$, using as few torches as possible.
Light levels are integers $l\in\mathbb{N}_0$, where $0$ indicates complete darkness.
If a torch is placed on a tile, the tile itself receives a light level of $\Ltorch$.
The light spreads to surrounding tiles and decreases gradually with distance $d$ described in \cref{eq:distance}, both horizontally and with elevation.
If a tile is illuminated by multiple nearby torches, it assumes the maximum over all light levels.
Note that the way that light spreads in Minecraft does not at all align with physical reality, but is game-specific.
Where in reality we would expect light to move in straight lines and be blocked by obstacles, light in Minecraft has no sense of direction, but spreads evenly to all neighboring blocks, and can even ``creep around corners''.
Let $\cT\subseteq\cS$ be the set of all tile sites that have a torch.
The light level of each tile $s$ can then be described mathematically as \begin{equation*}
	l(s~\vert ~\cT,\Ltorch) = \max_{t\in \cT} ~\max\lbrace 0,\Ltorch-d(s,t)\rbrace\;.
\end{equation*}
If $\cT=\emptyset$, then we set $\forall s: ~l(s~\vert ~\emptyset)=0$.
Overall, the problem we are trying to solve can be formalized as \begin{align}
	\text{min}\quad &\lvert \cT\rvert &&\text{with } \cT\subseteq\cS \label{eq:min-norm}\\
	\text{subject to}\quad &l(s~\vert ~\cT,\Ltorch)\geq \Lmin &&\forall s\in\cS\;. \label{eq:constr}
\end{align}
\noindent
In this article, we set $\Ltorch=14$ and $\Lmin=8$, which are the same values as in Minecraft.

\section{A QUBO formulation for \TP}
\label{sec:qubo}

As hinted at in the introduction to this article, \TP lends itself to \QUBO as its candidate solution space is the set of binary vectors $\BB^n$.
In order to derive a \QUBO formulation for \TP, we need to (1) embed the candidate solutions into binary vectors and (2) formulate a 2nd order polynomial whose minimizing binary vector corresponds to a solution of our original problem.
The latter is done by computing the \QUBO parameter values $Q_{ij}$ (see \cref{eq:qubo}) from our problem instance at hand, which in our case is a heightmap with tiles $\cS$ and its corresponding distance function $d:\cS\times\cS\rightarrow\mathbb{N}_0$ as described in \cref{sec:problem}.

We firstly arrange the elements of $\cS$ in an arbitrary but fixed order $S$, such that $S_i$ denotes the location of the $i$-th tile for all $i\in [n]$ with $n\defeq\abs{\cS}$.
A subset $\cT\subseteq\cS$ then corresponds to a binary vector $\vec x\in\BB^n$ via $\cT\equalhat(\ind{S_i\in\cT})^\top_{i\in[n]}$, where $\ind{\cdot}$ is the indicator function defined as \begin{equation*}
	\ind{\mathsf{B}} \defeq \begin{cases}
		1 &\text{if } \mathsf{B}\equiv\top,\\
		0 &\text{otherwise}\;.
	\end{cases}
\end{equation*}%
\noindent
for some Boolean expression $\mathsf{B}$.
\Cref{eq:min-norm} can now be written as \begin{align*}
	\text{min}\quad &\norm{\vec x}_1 \quad\text{with } \vec x\in\BB^n\;.
\end{align*}
It is easy to see that if we set $Q_{ii}=P ~\forall i\in[n]$ with some arbitrary positive value $P>0$, the energy function $f_{\mat Q}$ will only be minimal for the binary vector $\vec 0$ with $f_{\mat Q}(\vec 0)=0=\norm{\vec 0}_1$.

Embedding the constraints from \cref{eq:constr} into \QUBO is far more challenging:
Given a candidate solution $\vec x$, we need to ensure that the light level is at least $\Lmin$, which is equivalent to saying that the distance to the nearest torch must be short enough that the light coming from it is still bright enough.
To ensure this, the torch's distance can be at most $\Ltorch-\Lmin$: \begin{align}\label{eq:max-constr}
	\min_{\substack{j\in [n],\\ x_j=1}} d(S_i,S_j) &\leq \Ltorch-\Lmin ~\forall i\in[n]\;.
\end{align}

There are two problems with this set of inequalities:
Firstly, $\max$ is a non-linear operation, and \QUBO can only represent quadratic energy functions.
Secondly, \QUBO is unconstrained by definition, which is why we need to employ techniques to embed the constraints directly into the weight matrix $\mat Q$ such that the unconstrained solution still obeys our constraints.
In the following subsections we describe how we approached both of these challenges.

\subsection{Eliminating the Maximum Function}
\label{sec:eliminating}

There are several smooth approximations of the $\max$ functions, such as the LogSumExp, $p$-norm or the Boltzmann operator.
As an example, consider the first-mentioned function \begin{equation*}
	\mathrm{LSE}_{\alpha}(a_1,\dots,a_n)\defeq \frac{1}{\alpha}\log\sum_{i=1}^ne^{\alpha a_i}\;.
\end{equation*}
We can re-write our inequality constraints in \cref{eq:max-constr} by inverting the sign and incorporating the condition $x_j=1$ to $\max_{j\in[n]}((P_i-d(S_i,S_j))x_j-P_i)\geq \Lmin-\Ltorch$ for some $P_i>\min_{j\neq i}d(S_i,S_j)$, which using $\mathrm{LSE}_\alpha$ yields \begin{align*}
	&&\frac{1}{\alpha}\log\sum_je^{\alpha(P_i-d(S_i,S_j))x_j-\alpha P_i} &\geq \Lmin-\Ltorch \\
	\Leftrightarrow &&\sum_j e^{-\alpha P_i}e^{\alpha(P_i-d(S_i,S_j))x_j} &\geq e^{\alpha(\Lmin-\Ltorch)}
\end{align*}%
Thanks to all $x_j$ being binary, we can re-write $e^{ax_j}$ as $(e^a-1)x_j+1$, which yields a linear constraint \begin{align*}
	&&n+\sum_j e^{-\alpha P_i}(e^{\alpha(P_i-d(S_i,S_j))}-1)x_j &\geq e^{\alpha(\Lmin-\Ltorch)} \\
	\Leftrightarrow &&\vec x^\top\vec u^{(i)}+v_i&\leq 0\;,
\end{align*}%
where $u^{(i)}_j=e^{-\alpha P_i}-e^{-\alpha d(S_i,S_j)}$ and $v_i=e^{\alpha(\Lmin-\Ltorch)}-n$.
$\mathrm{LSE}_\alpha$ approaches $\max$ better the larger we choose $\alpha$.
We implemented these constraints and found that approximations of $\max$ lead to numerical instability due to the large magnitude differences between the coefficients, rendered solutions output from the \QUBO solvers useless.

However, there is a simpler formulation that exploits the binary nature of our constraints. For that, we define the matrix $\vec{D}\in\{0,1\}^{n\times n}$ with entries $d_{ij}$ to be
\begin{equation*}
	d_{ij}=\ind{d(S_i,S_j)\leq \Ltorch-\Lmin},\quad \forall i,j\in[n]\;.
\end{equation*}
In order to ensure that $\min_j d(S_i,S_j) \leq \Ltorch-\Lmin$, we can check whether \begin{equation}
	\sum_j x_jd_{ij}=\left(\vec{D}\vec{x}\right)_i \geq 1\;, 
	\label{eq:binary_distance}
\end{equation}
which immediately yields linear constraints and circumvents the $\max$ function entirely.

Using \cref{eq:binary_distance} the optimization problem in \cref{eq:min-norm,eq:constr} can be reformulated to 
\begin{align}
	\min_{\vec x\in\{0,1\}^n}\quad &\vec{1}^{\top}\vec{x} \label{eq:min_torches}\\
	\text{subject to}\quad &\vec{D}\vec{x}\geq \vec{1} \label{eq:min_light}\;,
\end{align}
where $\vec{1}$ is the $n$-dimensional vector consisting only of ones, and $\geq$ is applied element-wise.
Realizing this, we find that \TP has striking similarity to the \SC problem:

\begin{theorem}
	\TP is a special case of \SC.
\end{theorem}
\begin{proof}
	In \SC, we are given a set $A$ and a collection of subsets $B_i$ with $B_i\subseteq A$ for each $i\in I$, and $\bigcup_{i\in I}B_i=A$.
	The objective is to find a $J\subseteq I$ such that $\abs{J}$ is minimal and $\bigcup_{j\in J}B_j=A$.
	Given a heightmap with tiles $\cS$ and distance function $d$ as defined in \cref{eq:distance}, we set $A=\cS$, $I=[n]$ and $B_i=\lbrace s\in\cS: ~d(S_i, s)\leq\Ltorch-\Lmin\rbrace$ for all $i\in[n]$.
	Thus, a solution $J\subseteq I$ with minimal $\abs{J}$ is a minimal set of torches that illuminates all other tiles.
\end{proof}

\subsection{Handling Inequality Constraints}

The question remains how to obtain a \QUBO from the constrained formulation in \cref{eq:min_torches,eq:min_light}. 
To answer this question, we remark that the linear inequality constraint in \cref{eq:min_light} can be reformulated with using an auxiliary vector $\vec{z}\in\mathbb N_0^n$ with positive integer entries
\begin{align*}
	\vec{D}\vec{x}\geq \vec{1}\ \Leftrightarrow \
	\vec{D}\vec{x}-\vec{1}=\vec{z}\;,
\end{align*}
leading to the slightly different but equivalent problem formulation
\begin{align}
	\min_{\vec x\in\{0,1\}^n}\quad &\vec{1}^{\top}\vec{x} \label{eq:min_torches_trans}\\
	\text{subject to}\quad &\vec{D}\vec{x}-\vec{1}-\vec{z}= 0,\quad \vec{z}\in\mathbb N_0^n \label{eq:min_light_trans}\;.
\end{align}
A common way of solving an optimization problem with equality constraints as in \cref{eq:min_torches_trans,eq:min_light_trans} is to introduce a matrix of binary slack variables $\mat{S}\in \{0,1\}^{n\times m}$ with $m\defeq \left\lceil\log_2n\right\rceil$ \cite{vyskovcil2019}.
Considering the vector $\vec{r}$ corresponding to all powers of two up to $2^m$, i.e., $\vec{r}=\left(1, 2,4,\dots, 2^m\right)^{\top}$, we find that $\vec{z}=\mat{S}\vec{r}$.
The constrained problem in \cref{eq:min_torches_trans,eq:min_light_trans} can then be reformulated to an equivalent unconstrained problem by introducing a penalty parameter $\beta>0$
\begin{align}
	\min_{\vec x\in\{0,1\}^n,\mat{S}\in\{0,1\}^{n\times m}}\quad &\vec{1}^{\top}\vec{x}+\beta \left\|\mat{D}\vec{x}-\vec{1}-\mat{S}\vec{r}\right\|^2\;.
	\label{eq:slack_qubo}
\end{align}
Using this formulation, however, comes with the overhead of introducing $nm$ additional binary variables leading to a total \QUBO dimension of $n(m+1)$.
As we are still in the era of noisy intermediate-scale quantum (NISQ) computers \cite{preskill2018}, it is better to resort to a \QUBO formulation that uses fewer qubits.

To this end, we use an iterative method called Alternating Direction Method of Multipliers (ADMM) \cite{boyd2011}.
Firstly, we establish a new problem formulation
\begin{align}
	\min_{\vec x\in\{0,1\}^n,\vec{z}\in\mathbb Z^n}\quad &\vec{1}^{\top}\vec{x}+\gamma\vec{1}^{\top}\Theta(\vec{z})\\
	\text{subject to}\quad& \vec{c}\left(\vec{x},\vec{z}\right)= 0\;,
	\label{eq:problem_heaviside}
\end{align}
with $\gamma>0$, $\vec{c}\left(\vec{x},\vec{z}\right)\defeq \vec{D}\vec{x}-\vec{1}-\vec{z}$, and $\Theta$ being an element-wise step function defined as $\Theta(\vec z)\defeq(\ind{z_1< 0},\dots,\ind{z_n< 0})$.
The term $\vec{1}^{\top}\Theta(\vec{z})$ penalizes vectors $\vec{z}\in\mathbb Z^n$ with negative entries, since we want $\vec{z}\in\mathbb N^n_0$.
Note that optimal solution vectors $\vec{x}\in\{0,1\}^n$, $\vec{z}\in\mathbb Z^n$ of \cref{eq:problem_heaviside} are also optimal for the problem in \cref{eq:min_torches_trans,eq:min_light_trans}.
We then introduce the augmented Lagrangian \cite{powell1969, gabay1976}
\begin{equation}
	L\left(\vec{x},\vec{z},\vec{\lambda},\mu\right)\defeq \vec{1}^{\top}\vec{x}+\gamma\vec{1}^{\top}\Theta(\vec{z})+\vec{\lambda}^{\top}\vec{c}\left(\vec{x},\vec{z}\right)+\frac{\mu}{2}\left\|\vec{c}\left(\vec{x},\vec{z}\right)\right\|^2\;,
	\label{eq:aug_lagrangian}
\end{equation}
with $\vec{\lambda}$ and $\mu$ being coefficients and multipliers for the penalty terms.
For minimizing \cref{eq:aug_lagrangian}, we use ADMM, which is outlined in \cref{alg:admm}.
\begin{algorithm}[t]
	\caption{Alternating Direction Method of Multipliers (ADMM)}\label{alg:admm}
	\begin{algorithmic}[1]
		\Require Initial $\mu_0 > 0$,
		\Ensure $\vec{x}^*$, $\vec{z}^*$, $\vec{\lambda}^*$, $\mu^*$ optimizing $L\left(\vec{x},\vec{z},\vec{\lambda},\mu\right)$ 
		\State $\vec{z}^* \gets \vec{0}$
		\State $\vec{\lambda}^* \gets \vec{0}$
		\State $\mu^* \gets \mu_0$
		\Repeat
		\State $\vec{x}^* \gets \argmin_{\vec{x}} L\left(\vec{x},\vec{z}^*,\vec{\lambda}^*,\mu^*\right)$ 
		\Comment{Quantum annealer can be used}
		\label{alg:admm:qubo}
		\State $\vec{z}^* \gets \argmin_{\vec{z}} L\left(\vec{x}^*,\vec{z},\vec{\lambda}^*,\mu^*\right)$
		\label{alg:admm:z}
		\State $\vec{\lambda}^*\gets \vec{\lambda}^*+\mu^*\vec{c}\left(\vec{x},\vec{z}\right)$
		\State Update $\mu^*$ \label{alg:admm:update}
		\Until{a convergence criterium is met}
		\label{alg:admm:criterium} 
	\end{algorithmic}
\end{algorithm}
The vectors $\vec{x}$ and $\vec{z}$ are updated in an alternating fashion.
In our case, line \ref{alg:admm:qubo} can be written as
\begin{align*}
	\argmin_{\vec{x}\in\{0,1\}^n}\ L\left(\vec{x},\vec{z},\vec{\lambda},\mu\right)
	&=\argmin_{\vec{x}\in\{0,1\}^n}\ \vec{1}^{\top}\vec{x}+\vec{\lambda}^{\top}\vec{c}\left(\vec{x},\vec{z}\right)+\frac{\mu}{2}\left\|\vec{c}\left(\vec{x},\vec{z}\right)\right\|^2 \\
	&=\argmin_{\vec{x}\in\{0,1\}^n}\ \vec{1}^{\top}\vec{x}+\vec{\lambda}^{\top}\vec{D}\vec{x}+\frac{\mu}{2}\left(\vec{x}^{\top}\vec{D}^{\top}\vec{D}\vec{x}-\left(\vec{1}+\vec{z}\right)^{\top}\vec{D}\vec{x}\right)\;,
\end{align*}
corresponding to a \QUBO formulation in \cref{eq:qubo}, which can be solved with a quantum computer.
Line \ref{alg:admm:z} can be reduced to
\begin{align*}
	\argmin_{\vec{z}\in\mathbb Z^n}\ L\left(\vec{x},\vec{z},\vec{\lambda},\mu\right)
	&=\argmin_{\vec{z}\in\mathbb Z^n}\ \gamma\vec{1}^{\top}\Theta(\vec{z})+\vec{\lambda}^{\top}\vec{c}\left(\vec{x},\vec{z}\right)+\frac{\mu}{2}\left\|\vec{c}\left(\vec{x},\vec{z}\right)\right\|^2 \\
	&=\argmin_{\vec{z}\in\mathbb Z^n}\ \gamma\vec{1}^{\top}\Theta(\vec{z})+\frac{\mu}{2}\left\|\vec{D}\vec{x}-\vec{1}-\vec{z}+\frac{1}{\mu}\vec{\lambda}\right\|^2 \\
	&=\max\{\vec{0}, \vec{D}\vec{x}-\vec{1}\}\;,
\end{align*}
where $\vec{0}$ is the $n$-dimensional vector consisting only of zeros and $\max$ is taken element-wise.

For solving our original problem in \cref{eq:min_torches,eq:min_light} we solve a sequence of \QUBO problems and update the other parameters according to \cref{alg:admm}. 
Using a quantum computer for the \QUBO instances results in a hybrid quantum-classical algorithm. 
Let the subscript $k$ denote the state of the parameters $\vec{x}^*,\vec{z},^*\vec{\lambda}^*,\mu^*$ after the $k$-th iteration of ADMM.
We choose the following update rule for $\mu_k^*$ in line \ref{alg:admm:update}
\begin{align*}
	\mu_{k+1} = \begin{cases}
		\rho \mu_{k}, &\text{if }\left\|\vec{c}\left(\vec{x}_k,\vec{z}_k\right)\right\| > 10\mu_{k} \left\|\vec{D}\left(\vec{z}_k-\vec{z}_{k+1}\right)\right\|\;, \\
		\mu_{k}/\rho, &\text{if } \left\|\vec{D}\left(\vec{z}_k-\vec{z}_{k+1}\right)\right\|> 10\mu_{k}\left\|\vec{c}\left(\vec{x}_k,\vec{z}_k\right)\right\|\;,  \\
		\mu_k,&\text{else}\;,
	\end{cases}
\end{align*}
with a fixed learning rate $\rho\ge1$, following \cite{boyd2011}. 
In place of a convergence criterium in line \ref{alg:admm:criterium} we use a fixed budget of $N$ of calls to the quantum computer.

\section{Experimental Evaluation}
\label{sec:experiments}

\begin{figure}
	\centering
	\begin{subfigure}[c]{0.49\textwidth}
		\centering
		\includegraphics[width=\textwidth]{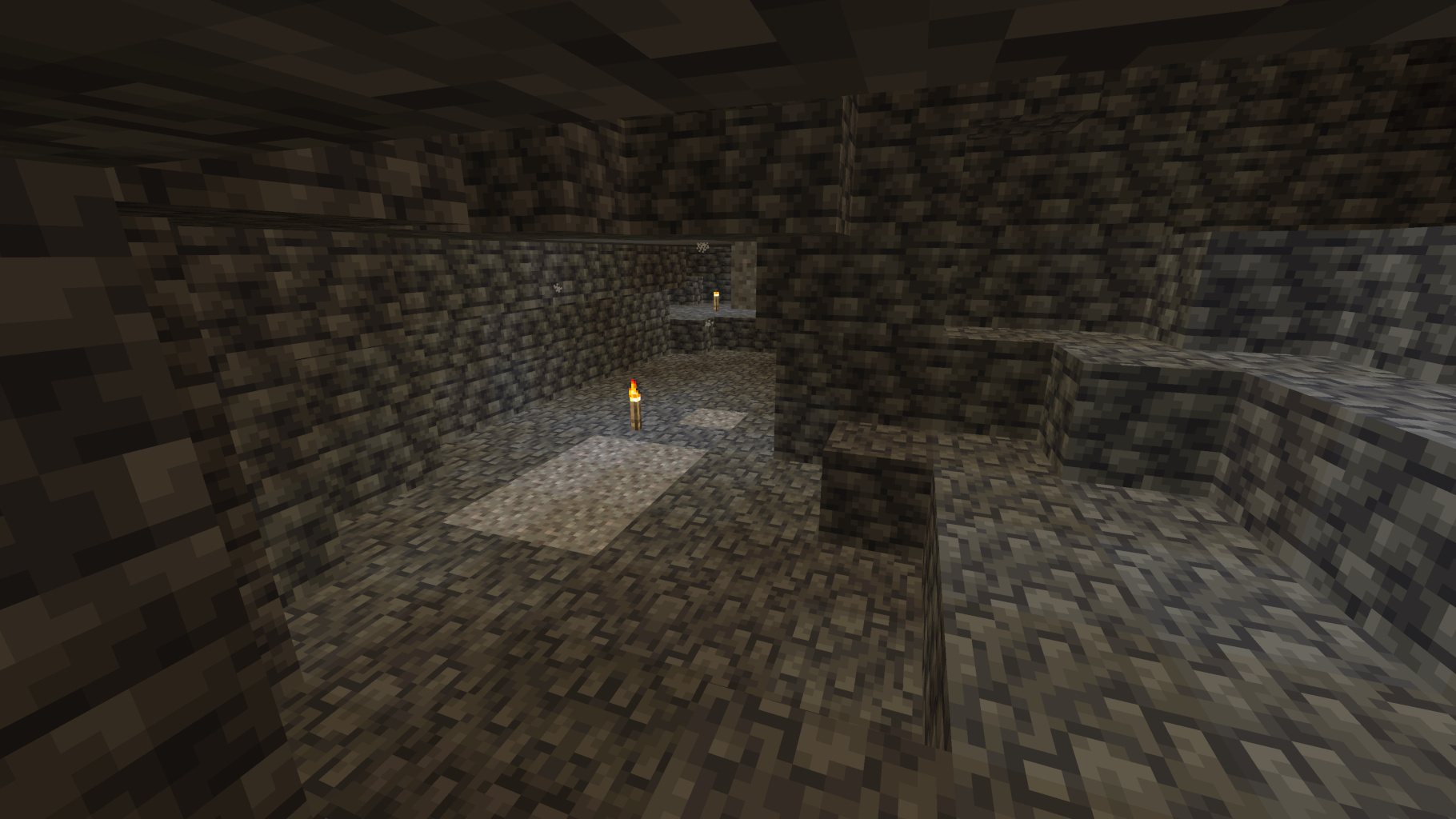}
	\end{subfigure}%
	\hfill
	\begin{subfigure}[c]{0.49\textwidth}
		\centering
		\includegraphics[width=\textwidth]{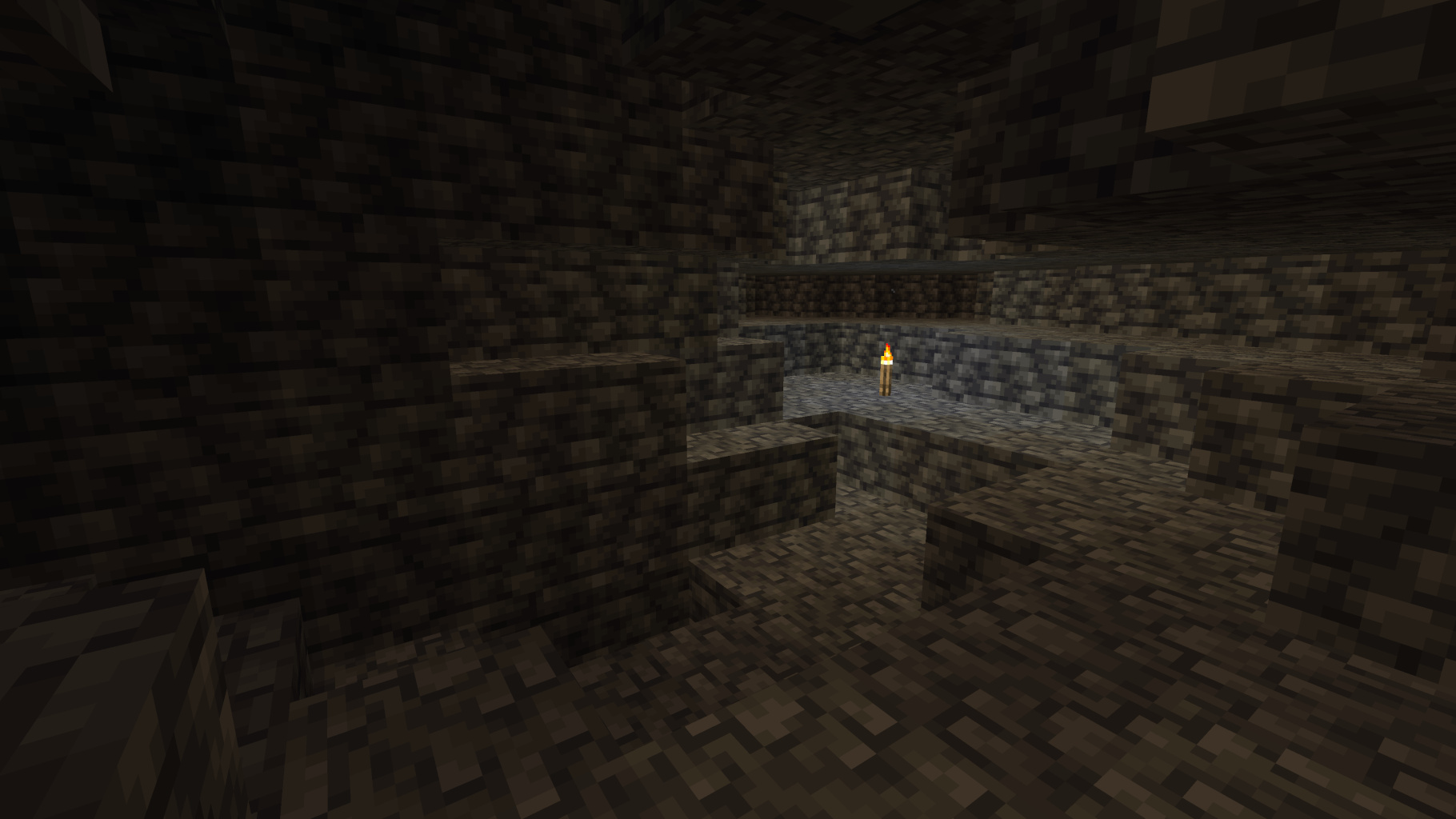}
	\end{subfigure}%
	\caption{Screenshots of Minecraft caves. In \cref{sec:experiments} we use heightmaps based on these caves for our experiments.}
	\label{fig:real_caves}
\end{figure}

We conduct experiments on different exemplary heightmaps.
Some heightsmaps are randomly generated from 2D Perlin noise \cite{perlin1985}, others are created from cave sections within a Minecraft world.
We compare four different \QUBO solvers utilizing the D-Wave Ocean\footnote{\url{https://docs.ocean.dwavesys.com/en/stable/}} Python package:
Simulated annealing (SA), Tabu search (Tabu), a combination of those two (TabuSA) and a real quantum annealer (QA), namely a D-Wave Advantage System 5.4 with 5614 qubits and 40,050 couplers.
The combination of SA and Tabu is achieved by using the D-Wave Hybrid\footnote{\url{https://docs.ocean.dwavesys.com/projects/hybrid/en/stable/}} Python package, parellelizing both methods, leading to better performance.
For all solvers we use the standard parameters.
As a pre-processing step we perform parameter compression \cite{mucke2023optimum}, which leads to improved results for QA.
The ADMM is repeated 10 times for every solver and heightmap and we plot the mean performance along with the $95\%$-confidence intervals.
We found the hyperparameters $\mu_0=0.01$ and $\rho=1.1$ to yield the best results, hence we fix these values for the upcoming experiments.

In \cref{fig:dwave_performance}, we compare the performances of the different \QUBO solvers for one real Minecraft cave from \cref{fig:real_caves} (right) with $n=67$.
\begin{figure}
	\centering
	\includegraphics[width=.9\textwidth]{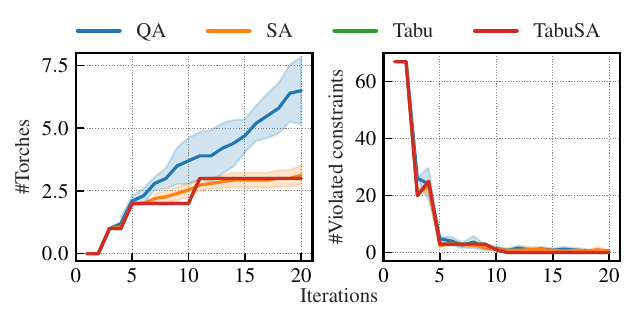}
	\caption{Performance of the four \QUBO solvers QA, SA, Tabu and TabuSA on the real Minecraft cave depicted in \cref{fig:real_caves} (bottom).}
	\label{fig:dwave_performance}
\end{figure}
For this, we plot the number of torches as well as the number of violated constraints over the iterations of ADMM.
We can see that the four \QUBO solvers converge to fulfilling all constraints, while QA places the most torches.
This indicates that its solution quality falls behind the other methods, which is a common problem of NISQ devices.
However, in \cref{fig:minecraft-cave-2} we depict the result of ADMM with QA for a particular run, after $2$, $5$ and $10$ iterations.
After 10 iterations, an optimal solution is found, such that every block is lit.

From now on, we use TabuSA, since the dimension of the upcoming problems is too large to obtain useful results with QA.
\cref{fig:minecraft-cave-1} depicts solutions after a certain number of ADMM iterations on the heightmaps of the other real Minecraft cave from \cref{fig:real_caves} (left) with $n=195$.
\begin{figure}
	\centering
	\begin{subfigure}[t]{0.24\textwidth}
		\centering
		\includegraphics[width=.7\textwidth]{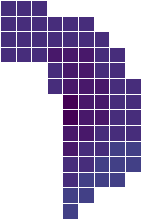}
		\subcaption{Heightmap}
	\end{subfigure}
	\begin{subfigure}[t]{0.24\textwidth}
		\centering
		\includegraphics[width=.7\textwidth]{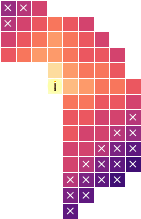}
		\caption{2 iterations, (20)}
	\end{subfigure}
	\begin{subfigure}[t]{0.24\textwidth}
		\centering
		\includegraphics[width=.7\textwidth]{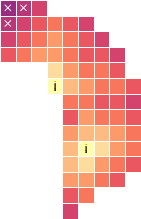}
		\caption{5 iterations (3)}
	\end{subfigure}
	\begin{subfigure}[t]{0.24\textwidth}
		\centering
		\includegraphics[width=.7\textwidth]{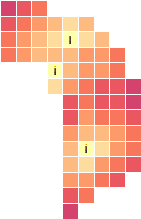}
		\caption{10 iterations (0)}
	\end{subfigure}
	\caption{Real Minecraft cave (\cref{fig:real_caves}, right) with $n=67$; white crosses signify light level below $\Lmin$. Number of constraint violations in parentheses. The \QUBO instances were solved using QA.}
	\label{fig:minecraft-cave-2}
\end{figure}
Again, we can see that with more iterations, more torches are placed and the constraints become fulfilled.
\begin{figure}
	\centering
	\begin{subfigure}[t]{0.24\textwidth}
		\centering
		\includegraphics[width=\textwidth]{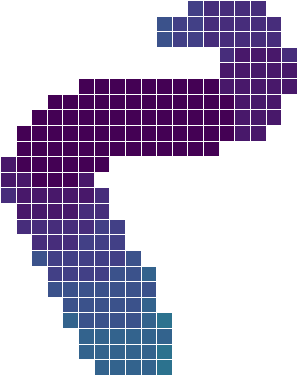}
		\subcaption{Heightmap}
	\end{subfigure}
	\begin{subfigure}[t]{0.24\textwidth}
		\centering
		\includegraphics[width=\textwidth]{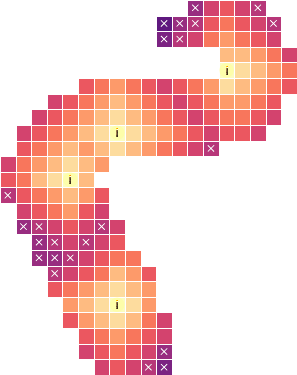}
		\caption{2 iterations (23)}%4 torches placed after 2 iterations, 23.}
\end{subfigure}
\begin{subfigure}[t]{0.24\textwidth}
	\centering
	\includegraphics[width=\textwidth]{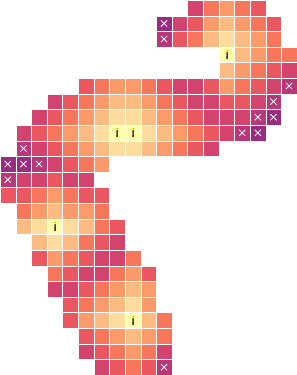}
	\caption{5 iterations (14)}%5 torches placed after 5 iterations, 14.}
	\end{subfigure}
	\begin{subfigure}[t]{0.24\textwidth}
\centering
\includegraphics[width=\textwidth]{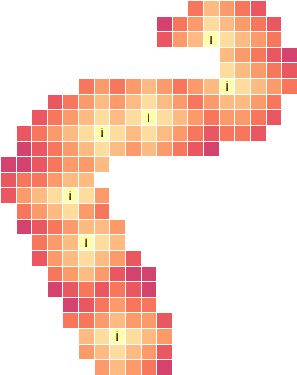}
\caption{After 20 iterations (0)}%7 torches placed after 20 iterations, 0.}
\end{subfigure}
\caption{Real Minecraft cave (\cref{fig:real_caves}, left) with $n=195$; white crosses signify light level below $\Lmin$. Number of constraint violations in parentheses. The \QUBO instances were solved with TabuSA.}
\label{fig:minecraft-cave-1}
\end{figure}
In \cref{fig:20x15_02a,fig:20x15_07a,fig:40x30_00a,fig:40x30_01a} we depict ADMM solutions for randomly generated heightmaps of varying sizes, using the TabuSA solver.
For every map, two intermediate results during the Lagrangian learning phase are shown.
We can see that fewer and fewer constraints are violated, until finally all tiles are lit with a small number of torches.

\begin{figure}[t]
	\centering
	\begin{subfigure}[t]{0.24\textwidth}
		\centering
		\includegraphics[width=\textwidth]{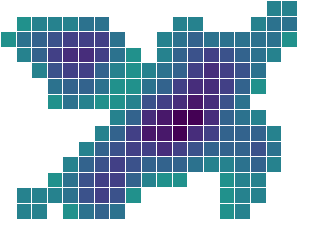}
		\subcaption{Heightmap}
	\end{subfigure}
	\begin{subfigure}[t]{0.24\textwidth}
		\centering
		\includegraphics[width=\textwidth]{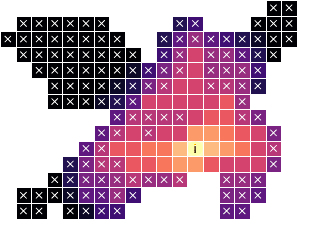}
		\caption{2 iterations (128)}
	\end{subfigure}
	\begin{subfigure}[t]{0.24\textwidth}
		\centering
		\includegraphics[width=\textwidth]{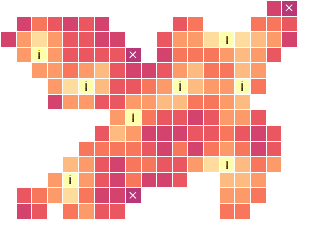}
		\caption{10 iterations (3)}
	\end{subfigure}
	\begin{subfigure}[t]{0.24\textwidth}
		\centering
		\includegraphics[width=\textwidth]{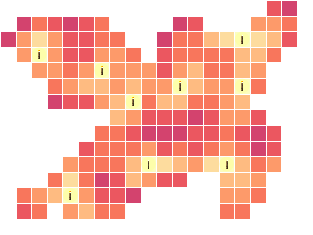}
		\caption{20 iterations (0)}
	\end{subfigure}
	\caption{Generated Minecraft cave with $n=168$; white crosses signify light level below $\Lmin$. Number of constraint violations in parentheses.}
	\label{fig:20x15_02a}
\end{figure}

\begin{figure}[t]
	\centering
	\begin{subfigure}[t]{0.24\textwidth}
		\centering
		\includegraphics[width=\textwidth]{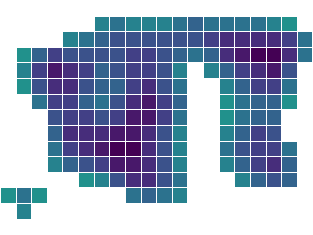}
		\subcaption{Heightmap}
	\end{subfigure}
	\begin{subfigure}[t]{0.24\textwidth}
		\centering
		\includegraphics[width=\textwidth]{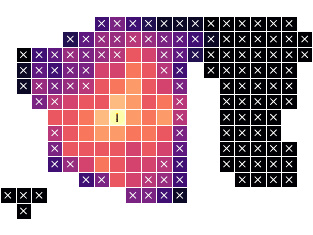}
		\caption{2 iterations (122)}
	\end{subfigure}
	\begin{subfigure}[t]{0.24\textwidth}
		\centering
		\includegraphics[width=\textwidth]{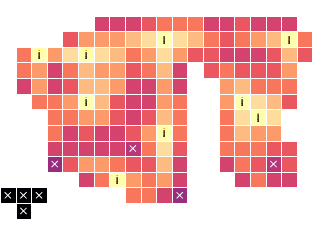}
		\caption{10 iterations (8)}
	\end{subfigure}
	\begin{subfigure}[t]{0.24\textwidth}
		\centering
		\includegraphics[width=\textwidth]{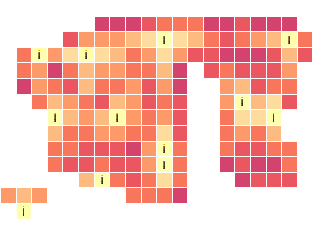}
		\caption{20 iterations (0)}
	\end{subfigure}
	\caption{Generated Minecraft cave with $n=169$; white crosses signify light level below $\Lmin$. Number of constraint violations in parentheses.}
	\label{fig:20x15_07a}
\end{figure}

\begin{figure}[t]
	\centering
	\begin{subfigure}[t]{0.24\textwidth}
		\centering
		\includegraphics[width=\textwidth]{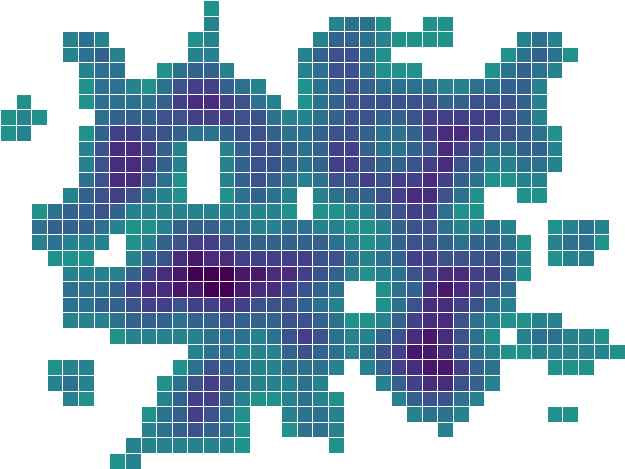}
		\subcaption{Heightmap}
	\end{subfigure}
	\begin{subfigure}[t]{0.24\textwidth}
		\centering
		\includegraphics[width=\textwidth]{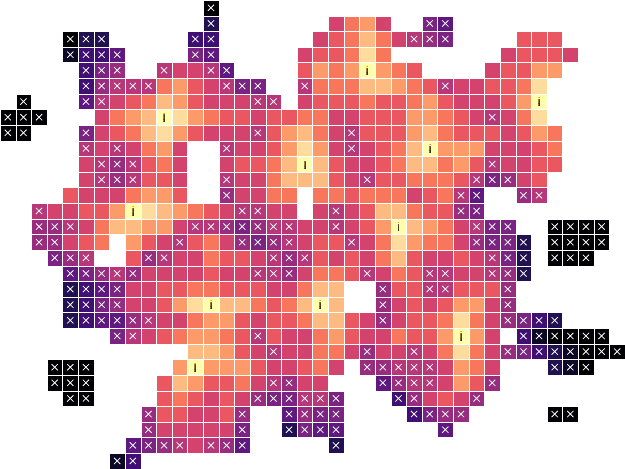}
		\caption{2 iterations (243)}
	\end{subfigure}
	\begin{subfigure}[t]{0.24\textwidth}
		\centering
		\includegraphics[width=\textwidth]{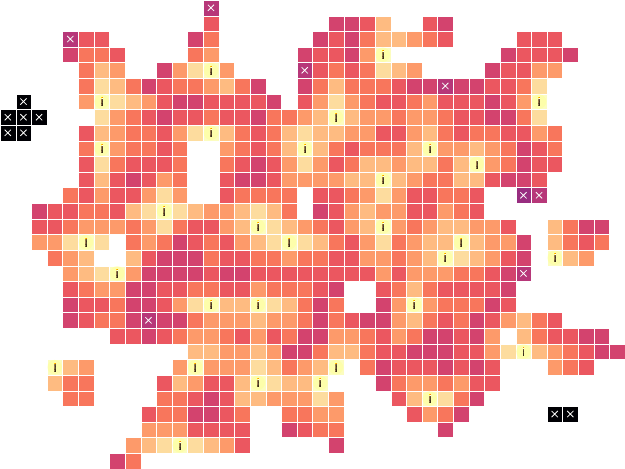}
		\caption{10 iterations (16)}
	\end{subfigure}
	\begin{subfigure}[t]{0.24\textwidth}
		\centering
		\includegraphics[width=\textwidth]{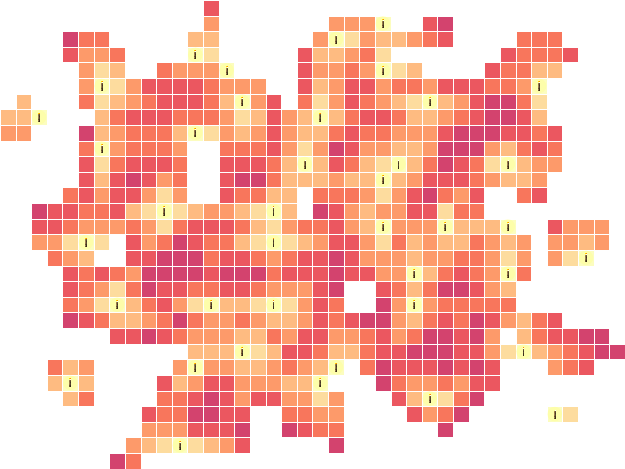}
		\caption{30 iterations (0)}
	\end{subfigure}
	\caption{Generated Minecraft cave with $n=711$; white crosses signify light level below $\Lmin$. Number of constraint violations in parentheses.}
	\label{fig:40x30_00a}
\end{figure}

\begin{figure}[t]
	\centering
	\begin{subfigure}[t]{0.24\textwidth}
		\centering
		\includegraphics[width=\textwidth]{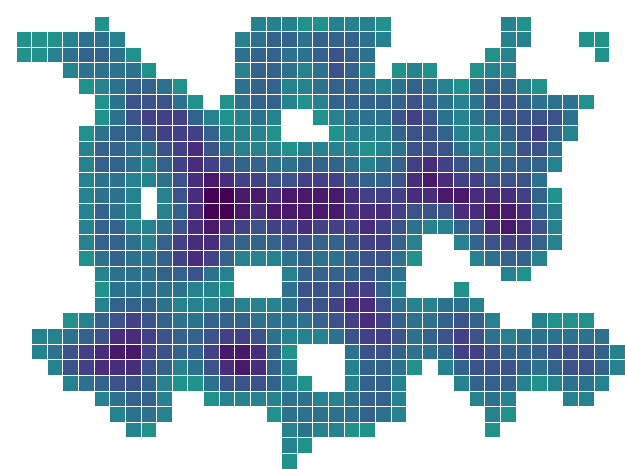}
		\subcaption{Heightmap}
	\end{subfigure}
	\begin{subfigure}[t]{0.24\textwidth}
		\centering
		\includegraphics[width=\textwidth]{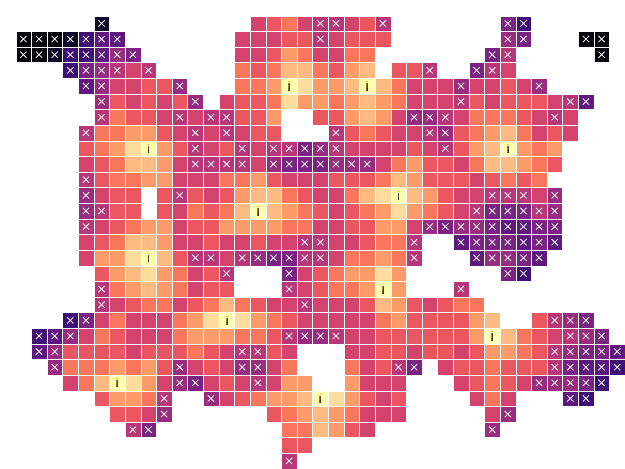}
		\caption{2 iterations (243)}
	\end{subfigure}
	\begin{subfigure}[t]{0.24\textwidth}
		\centering
		\includegraphics[width=\textwidth]{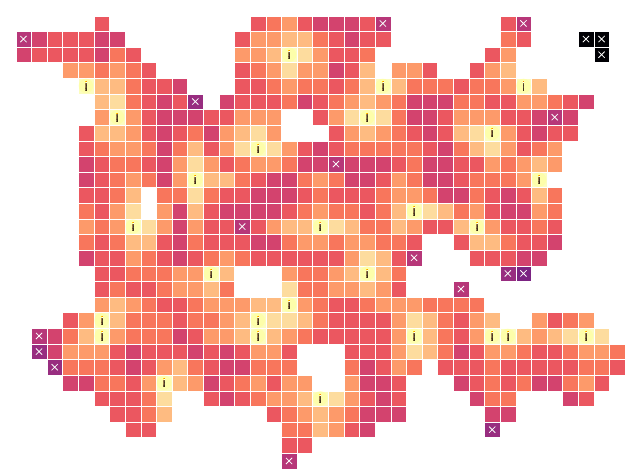}
		\caption{10 iterations (16)}
	\end{subfigure}
	\begin{subfigure}[t]{0.24\textwidth}
		\centering
		\includegraphics[width=\textwidth]{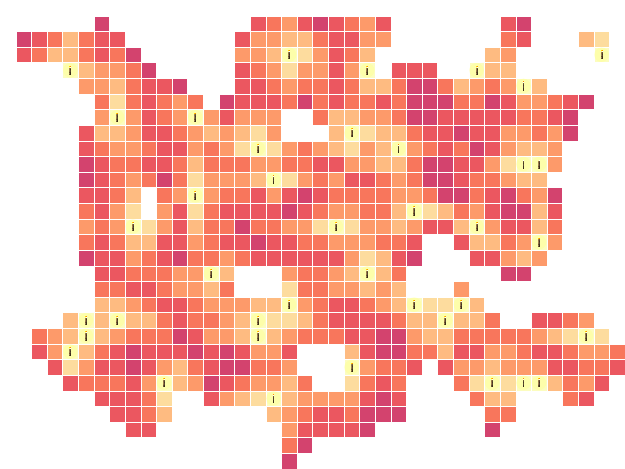}
		\caption{30 iterations (0)}
	\end{subfigure}
	\caption{Generated Minecraft cave with $n=708$; white crosses signify light level below $\Lmin$. Number of constraint violations in parentheses.}
	\label{fig:40x30_01a}
\end{figure}

\section{Related Work}
\label{sec:related}

Analyzing games in terms of their computational complexity has been part of computer science research for decades \cite{kaye2000,kendall2008,guala2014}.
Minecraft, in particular, has been used as a tool for research \cite{nebel2016}.

As mentioned in \cref{sec:eliminating}, it turns out that \TP is a special case of \SC, where the subsets have a spatial interpretation in discrete 3D space.
A \QUBO formulation of \SC is given in \cite{lucas2014}, where auxiliary qubits are used to encode inequality constraints.
We chose a different strategy by employing ADMM \cite{boyd2011} to learn the Lagrange multipliers $\vec \lambda$ iteratively in a hybrid quantum-classical fashion, which leads to smaller \QUBO instances with fewer qubits.
A similar approach is followed in \cite{djidjev2023} using a custom iterative Lagrangian algorithm, which we tried and found not to yield good results with our inequality constraints.
ADMM, which we chose instead, is more firmly grounded in theory.

We also find that \TP and \SC are similar to other problems such as \textsc{MaximumCoverage}, where an upper limit $k$ is given, and -- staying in the context of \TP -- we are looking to light up the largest possible area using at most $k$ torches.

\section{Conclusion}
\label{sec:conclusion}

In this article, we showed that the \TP problem arising from the video game Minecraft can be solved on quantum computers.
To this end, we used an ADMM-based procedure to learn Lagrange multipliers in a \QUBO formulation, which we then solved on a quantum annealer.
In contrast to other methods for solving the more general \SC problem we require no auxiliary variables.
The results demonstrated that the method works well for heightmaps with up to 700 floor tiles.
Due to limitations of today's noisy quantum hardware, we were able to solve the problem on an actual quantum annealer with heightmaps up to $n\leq 100$.

While we considered \TP from a video game perspective, it generalizes to similar setups, where we want to cover a certain area by using as few resources as possible:
Sensors or surveillance cameras with a fixed finite set of possible locations come to mind.
When we can define a distance between these locations, and we require a maximum distance from each location to the nearest resource, our method is applicable.

To the best of our knowledge, this article is the first instance of quantum computing being applied to solving a problem arising from Minecraft.
While this is a more lighthearted context to employ this relatively new computing resource, it has the potential to grant us insights that may be useful for applications to real-word problems, and serve as an illustrative example of how to approach problems from a quantum computing perspective.

\begin{acknowledgments}
    This research has been funded by the Federal Ministry of Education and
    Research of Germany and the state of North Rhine-Westphalia as part of the
    Lamarr Institute for Machine Learning and Artificial Intelligence.
\end{acknowledgments}

\end{document}